\documentclass[journal]{IEEEtran}

\usepackage{xcolor,soul,amssymb,amsthm,amsmath,dsfont} 
\usepackage{tikz} 
\usepackage{color}
\usepackage{graphicx,algorithm,algorithmic}
\graphicspath{{figures/}{./}}
\usepackage{array}
\usepackage{eqparbox}
\usepackage{url}
\usepackage{relsize}
\usepackage{stfloats}

\newcommand{\purple}{\color{purple}}

\newtheorem{theorem}{Theorem}

\newtheorem{problem}{Problem}

\newtheorem{corollary}[theorem]{Corollary}

\newtheorem{example}{Example}

\begin{document}
\bstctlcite{}
    \title{Promotion/Inhibition Effects in Networks:\\A Model with Negative Probabilities}
  \author{Anqi Dong,~\IEEEmembership{Student Member,~IEEE,}
                     Tryphon T. Georgiou,~\IEEEmembership{Life Fellow,~IEEE,}\\
                      and~Allen Tannenbaum,~\IEEEmembership{Fellow,~IEEE}
  \thanks{Anqi Dong and Tryphon T.~Georgiou are with the Department of Mechanical and Aerospace Engineering, University of California, Irvine, CA 92697; \{anqid2,tryphon\}@uci.edu.}
  \thanks{Allen Tannenbaum is with the Departments of Computer Science and Applied Mathematics \& Statistics at the State University of New York, Stony Brook; arobertan@gmail.com}}

\markboth{Dong \MakeLowercase{\textit{et al.}}: Negative Probabilities in Networks
}{Dong \MakeLowercase{\textit{et al.}}: Negative Probabilities in Networks}

\maketitle
 
\begin{abstract}
Biological networks often encapsulate promotion/inhibition as signed edge-weights of a graph. Nodes may correspond to genes assigned expression levels (mass) of respective proteins. The promotion/inhibition nature of co-expression between nodes is encoded in the sign of the corresponding entry of a sign-indefinite adjacency matrix, though the strength of such co-expression (i.e., the precise value of edge weights) cannot typically be directly measured. Herein we address the inverse problem to determine network edge-weights  based on a sign-indefinite adjacency and expression levels at the nodes. While our motivation originates in gene networks, the framework applies to networks where promotion/inhibition dictates a stationary mass distribution at the nodes. In order to identify suitable edge-weights we adopt a framework of ``negative probabilities,'' advocated by P.\ Dirac and R.\ Feynman, and we set up a likelihood formalism to obtain values for the sought edge-weights. The proposed optimization problem can be solved via a generalization of the well-known Sinkhorn algorithm; in our setting the Sinkhorn-type ``diagonal scalings'' are multiplicative or inverse-multiplicative, depending on the sign of the respective entries in the adjacency matrix, with value computed as the positive root of a quadratic polynomial.
\footnote{Research supported in part by the  AFOSR under grant FA9550-23-1-0096 and ARO under W911NF-22-1-0292. The authorship is alphabetical; the conception of this work is due to the second author.}
\end{abstract}


\IEEEpeerreviewmaketitle
\section{Introduction}

\IEEEPARstart{N}{etworks} are often used to encode interactions between chemical or biological compounds, such as proteins and genes \cite{barabasi2004network,sandhu2015graph,sandhu2016geometry,liu2020computational,seccilmics2020uncovering}. A distinguishing feature of such networks is that edge-weights quantify activation/suppression rates, equivalently,  promotion/inhibition relations between nodes. These in turn impact expression levels of substances assigned to node sites.
Measuring precisely the rates that regulate such interactions is often a difficult task, whereas measuring expression levels at the node sites is substantially easier. Thus, the subject of the present work is to develop a framework for solving the inverse problem to identify such network parameters (edge-weights) from information on expression levels and affinity between nodes. Specifically, we seek to {\em determine sign-indefinite values for edge-weights of a network based on knowledge of nodal mass and the sign of the edge-weights}. The sign of edge-weights is made available in the form of a sign-indefinite adjacency matrix. It is also possible that an estimate of (signed) edge-weights is already available and needs to be updated, so as to restore consistency with measured mass at the nodes. This situation is completely analogous to the case where the prior is only the signed adjacency matrix, and can be treated similarly.

We introduce a natural framework to identify sign-indefinite edge-weights from this type information, by adopting a sign-indefinite transition-probability model to encapsulate ``promotion vs.\ inhibition.'' Incorporating ``negative probabilities'' in a model was advocated early on
by Paul Dirac~\cite{dirac1942bakerian} and  Richard Feynman~\cite{feynman1987negative}, and 
in some scarce work that followed~\cite{bartlett1945negative,burgin2010interpretations}. Dirac and Feynman argued specifically that negative probabilities, when not directly and experimentally measurable, are perfectly acceptable as reflecting internal unobservable manifestations of an underlying mechanism. In our setting however, the interpretation of negative transition probabilities as rates is physically meaningful. Activation/suppression rates, reflected in sign-indefinite edge-weights, dictate the stationary distribution at the nodes (e.g., protein levels, and so on) and are sought to explain the observed data. 

The significance of identifying edge-weights in a gene regulatory network is of vital importance in that it allows assessing the role of nodes and edges in the overall functionality and robustness of networks \cite{sandhu2015graph,sandhu2016geometry}. Indeed, chemical deactivation of sites in a gene regulatory network that may contribute to selective cell-death has been the rationale behind many types of medical treatment of cancer, and the endeavour to properly quantify functionality and robustness of gene networks with suitable metrics, that point to vital nodes and links, is on-going \cite{baptista2023charting}. In particular, various centrality measures as well as network curvature have been proposed to quantify significance of nodes/edges as well as resilience of the network to changes. Determining the value of such metrics presupposes knowledge of edge-weights, which is the modest goal of the inverse problem that we consider herein.

A contribution of the present work is to propose identifying edge weights by minimizing {\em a suitable relative entropy functional between sign-indefinite measures}. The problem can be solved by a {\em Sinkhorn-like iteration} to obtain the minimizer. This new algorithm reduces to the well-known Sinkhorn ``diagonal-scaling iteration'' when the network adjacency matrix is sign definite. In the generality of the present model, with sign-indefinite edge weights, the algorithm amounts to an iterative scaling scheme with the scaling of multiplicative or inverse-multiplicative nature, depending on the sign of a corresponding entry in the adjacency matrix. At each step of the iteration, the scaling factor is readily available as the positive root of a quadratic polynomial.
The new algorithm, much like the Sinkhorn iteration, can be seen as effecting coordinate ascent to obtain the maximum of a concave functional, displaying as a consequence linear convergence rates.

Below, in Section \ref{sec:sec1} we present our problem formulation, followed by theoretical development and algorithmic considerations that are presented in Section \ref{sec:sec2}. Section \ref{sec:examples} provides illustrative examples.

\section{Problem formulation}\label{sec:sec1}
The data for our mathematical problem consist of two entities: first a symmetric matrix $A=A^\prime\in\mathbb \{-1,0,1\}^{n\times n}$, and then a (column) probability vector $p\in\mathbb R^n$, i.e., such that $p_i\geq 0$ for $i\in\{1,\ldots,n\}$ and $\sum_{i=1:n}p_i=1$.

As noted in the introduction, the matrix $A$ may represent the sign-indefinite adjacency of a given gene regulatory network. In such an example, the nodes  $v\in\mathcal V=\{1,2,\ldots,n\}$ correspond to genes and the vertex set cardinality $|\mathcal V|=n$ is typically very large, with $n$ of the order of $10^3-10^4$. The adjacency matrix $A=[A_{ij}]_{i,j=1:n}$ encodes both, a ``connectivity'' structure among genes as well as a constructive/indifferent/destructive (promoting/no effect/inhibiting) contribution of the corresponding pair of genes to the respective protein production levels. This information is encoded in the sign of the respective entry $A_{ij}\in\{-1,0,1\}$, with positive value indicating a constructive co-expression, zero indicating no perceived effect, and negative indicating inhibiting effect. This (biological) information is assumed given to us.

The second piece of datum for our problem is a probability vector
$p_i$ ($i\in\{1,\ldots,n\}$ that represents corresponding relative protein levels, and hence respective potency of the genes sites. That is, $p_i\geq 0$ and $\sum_{i=1:n} p_i=1$, with high values representing strong expression of respective proteins. Once again, this information is given to us from experimental data.

Biologists are typically interested in modeling the contributions of various genes on relative protein production levels in a quantitative manner. To this end, ad-hoc schemes have been utilized; see e.g., \cite{sandhu2016geometry,barabasi2004network} and the references therein. In earlier work, the rationale behind the proposed schemes was to employ the theory of Markov chains by calibrating the constructive/inhibiting effect of sites to production levels to only positive values. This was done by adding a suitable positive constant to all entries of $A$. Accordingly, $A$ becomes a positive matrix, to which the standard theory of Markov chains may be adapted, after scaling $A$ suitably, so that $p$ may be regarded as the stationary distribution of a corresponding Markov chain. Several pitfalls plague all such schemes, and are traceable to the difference between constructive and inhibiting effects, since adding a constant, evidently creates a bias in one direction and not the other.

In the present work, we propose an approach that seems natural for the problem at hand, and brings in the concept of {\em negative probabilities}. As indicated earlier, such concepts have had notable proponents including Dirac and Feynman, and some scant following.
In some detail, we hereby, postulate and seek a {\em sign-indefinite Markov transition} model to explain the observed invariant distribution in $p$, while acknowledging the promoting/inhibiting nature of the links between genes. Feynman's dictum suggests that, as long as internal probabilities in a model are not experimentally observable (as is the case with the sought transition kernel that will be constructed to respect the signs of $A_{ij}$'s), it is acceptable provided that it explains observed and experimentally measurable (non-negative) probabilities and relative frequencies.

Thus, we seek to identify a {\em sign-indefinite probability transition matrix} $\Pi=[\Pi_{ij}]_{i,j=1:n}$ such that $\Pi_{ij}\gtreqqless 0$ in accordance with the similar property for the corresponding $A_{ij}$. Alternatively, $\frac{\Pi_{ij}}{A_{ij}}\geq 0$, while $A_{ij}=0\Rightarrow \Pi_{ij}=0$. To this end, we formulate the following problem seeking to minimize a suitable functional (herein, entropy).

\begin{problem}\label{prob:prob1} Determine $\Pi$ that satisfies the above conditions,
    \begin{subequations}\label{eq:eqs}
     minimizes
    \begin{align}\label{eq:eq2}
       J(\Pi,A):= \!\!\! &\sum_{i,j|A_{i,j}\neq 0} \! \! p_i\frac{\Pi_{ij}}{A_{ij}}\log (\frac{\Pi_{ij}}{A_{ij}}),
    \end{align}
    and satisfies
   \begin{align}\label{eq:eq1a}
        &\sum_{i=1:n}p_i\Pi_{ij}=p_j, \mbox{ for all }j=1:n,\\
        \label{eq:eq1b}
        &\sum_{j=1:n}\Pi_{ij}=1, \mbox{ for all }i=1:n.
    \end{align}
    \end{subequations}
\end{problem}

The use of relative entropy between $A$ and $\Pi$, as above, echoes a similar usage in the justifications of Schr\"odinger's bridges, rooted in large deviations theory (see \cite{chen2021stochastic}. It has also been central in other problems in network theory \cite{zhou2021optimal}. Negative transition probabilities in our context can be interpreted as promotion/inhibition rates.

\section{Theoretical development}\label{sec:sec2}

We first reformulate Problem \ref{prob:prob1} in the following manner.  We define
\[
\mathbf A:=|A|:=[|A_{ij}|]_{i,j=1:n}.
\]
This is seen to be a sign-definite adjacency of the network with entries $\mathbf A_{i,j}\in\{0,1\}$. Likewise, for a sign-indefinite transition probability $\Pi$ as earlier, we write
\[
\mathbf{\Pi} =|\Pi|=[|\Pi_{ij}|]_{i,j=1:n}
\]
for the corresponding (typically) unnormalized transition probability matrix. That is, $\mathbf \Pi$ has non-negative entries, but without guaranteed row-sums being equal to one. Thus, our problem can be re-cast as follows.
\begin{problem}\label{prob:prob2} Determine an entry-wise nonnegative matrix $\mathbf{\Pi}$ that minimizes
\begin{subequations}\label{eq:eqs2}
\begin{align}\label{eq:prob3}
   \mathbf J(\mathbf \Pi,\mathbf A):= \!\!\! \sum_{i,j|\mathbf{A}_{i,j}\neq 0} \! \! p_i\frac{{\mathbf{\Pi}}_{ij}}{\mathbf{A}_{ij}}\log (\frac{{\mathbf{\Pi}}_{ij}}{\mathbf{A}_{ij}}) 
\end{align}
and satisfies the {\bf linear} constraints
\begin{align}
    \label{eq:prob3b}
    &\sum_i p_iA_{ij}\mathbf{\Pi}_{ij}=p_j\\
    \label{eq:prob3c}
    &\sum_jA_{ij}\mathbf{\Pi}_{ij}=1,
\end{align}
for $j,i\in\{1,\ldots,n\}$.
\end{subequations}
\end{problem}

It is clear that Problems \ref{prob:prob1} and \ref{prob:prob2} are equivalent, and of course, that $\mathbf J(\mathbf \Pi,\mathbf A)=J(\Pi,A)$. Problem \ref{prob:prob2} is quite similar to the classical Schr\"odinger problem \cite{chen2021stochastic}, but because $A$ is sign-indefinite, (\ref{eq:prob3b}-\ref{eq:prob3c}) are not the usual conditions. Yet, they are still linear, and the above re-formulation readily leads to the following conclusion.

\begin{theorem}\label{thm:thm1}
    If Problem \ref{prob:prob2} is feasible, the minimizer exists and is unique.
\end{theorem}

\begin{proof}
  The result follows by virtue of the fact that $\mathbf J(\mathbf \Pi,\mathbf A)$ is strictly convex in $\mathbf \Pi$, and the constraints are linear.
\end{proof}

\begin{corollary}\label{thm:cor1}
     If Problem \ref{prob:prob1} is feasible, the minimizer exists and is unique.
\end{corollary}

\subsection{A Sinkhorn-like algorithm}
We now seek to develop a computational approach for our problem that is akin to the Sinkhorn iteration~\cite{georgiou2015positive,chen2021stochastic}. To this end, we invoke duality theory so as to obtain the functional form of the minimizer. Thus, we introduce Lagrange multipliers and obtain the Lagrangian
\begin{align*}
{\mathcal L}(\mathbf{\Pi},p,\lambda,\mu):=
\! \! \! &\sum_{i,j|\mathbf{A}_{i,j}\neq 0} \!\!\!p_i\frac{{\mathbf{\Pi}}_{ij}}{\mathbf{A}_{ij}}\log (\frac{{\mathbf{\Pi}}_{ij}}{\mathbf{A}_{ij}})\\
&+ 
\sum_{j} \mu_j \big(\sum_i  p_iA_{ij} \mathbf{\Pi}_{ij}-p_j\big)\\
&+  
\sum_{i} \lambda_i \big(\sum_j A_{ij} \mathbf{\Pi}_{ij} - 1 \big).
\end{align*}
The first-order optimality condition, $\partial \mathcal L/\partial \mathbf \Pi_{ij}=0$, gives that
\begin{align*}
   \frac{p_i}{\mathbf A_{ij}} \log(\frac{\mathbf{\Pi}_{ij}}{\mathbf A_{ij}}) + \frac{p_i}{\mathbf A_{ij}} + \mu_j p_i A_{ij} + \lambda_i A_{ij}=0,
\end{align*}
for $i,j$  such that $\mathbf A_{ij}=1$, otherwise $\mathbf \Pi_{ij}=0$. Since $\mathbf A_{i,j}\in\{0,1\}$ and $\mathbf A_{i,j}A_{i,j}=A_{i,j}$, the optimizer must have the following functional dependence on the Lagrange multipliers
\begin{align}\label{eq:Pistar}
 \mathbf{\Pi}^*_{ij} = \mathbf A_{ij}\times\exp\big(- 1 - \mu_jA_{ij} -  \frac{\lambda_i}{p_i}A_{ij} \big).
\end{align}

Define $A^{+}=\frac12\left(\mathbf A+A\right)$ and $A^{-}=\frac12\left(\mathbf A-A\right)$, so that $\mathbf A =|A|= A^{+} + A^{-}$ and $A= A^+ - A^-$,  and a scaled set of new parameters $\nu_i:=\lambda_i/p_i$, for $i\in\{1,\ldots,n\}$. Then, the optimal kernel may be written in the form
\begin{align}\label{eq:pi}
    \mathbf{\Pi}^*_{ij} = 
\begin{cases}
    \exp\big(- 1 - \mu_j - \nu_i\big),  &\mbox{when }A_{ij}>0,\\
    \exp\big(- 1 + \mu_j + \nu_i\big)&\mbox{when }A_{ij}<0,\\
    0  &\mbox{when }A_{ij}=0,
\end{cases}
\end{align}
and equivalently, in the form
\begin{align}\label{eq:pi2}
    \mathbf{\Pi}^*_{ij} = \frac{1}{e}\left(A_{ij}^+ e^{-(\mu_j+\nu_i)}+A_{ij}^- e^{+(\mu_j+\nu_i)}\right),
\end{align}
with constraints $F_j(\mu_j)=p_j$ and $B_i(\nu_i)=1$, where
\begin{align*}
  F_j(\mu_j):= &\sum_{i=1:n} \frac{1}{e}\left(A_{ij}^+ e^{-(\mu_j+\nu_i)}-A_{ij}^- e^{+(\mu_j+\nu_i)}\right)p_i\\B_i(\nu_i):=&\sum_{j=1:n} \frac{1}{e}\left(A_{ij}^+ e^{-(\mu_j+\nu_i)}-A_{ij}^- e^{+(\mu_j+\nu_i)}\right).
\end{align*}
We re-write,
\begin{subequations}
\begin{align}\label{eq:F_j}
  F_j(\mu_j) &= a^F_j(\nu) e^{-\mu_j}
   -b^F_j(\nu)e^{\mu_j}\\
\label{eq:B_i}
  B_i(\nu_i) &= a^B_i(\mu) e^{-\nu_i}
   -b^B_i(\mu) e^{\nu_i},
\end{align} 
\end{subequations}
for
\begin{align*}
    &a^F_j(\nu):=\frac{1}{e}\sum_{i=1:n} 
   A_{ij}^+ p_i e^{-\nu_i}, &
   b^F_j(\nu):=\frac{1}{e}\sum_{i=1:n}A_{ij}^-p_i e^{\nu_i},\\
  &a^B_i(\mu) := \frac{1}{e}\sum_{j=1:n} 
   A_{ij}^+ e^{-\mu_j}, &
   b^B_i(\mu):=\frac{1}{e}\sum_{j=1:n}A_{ij}^- e^{\mu_j}, 
\end{align*} 
and we readily observe that $\mu_j$ can be computed explicitly from the constraint $F_j(\mu_j)=p_j$, when the vector $\nu$ is kept fixed, and similarly, $\nu_i$ can be computed from $B_i(\nu_i)=1$.

To see this, consider the function $f(x)=ae^{-x}-be^x$, with $a,b>0$, and observe that it is monotonic, with negative derivative on the whole real axis, having limits $f(-\infty)=\infty$ and $f(\infty)=-\infty$. Thus, for any given value $c$, a solution to $f(x)=c$ may be readily obtained as the logarithm of the positive root of a quadratic, giving,
\begin{subequations}\label{eq:logquad}
\begin{align}\label{eq:logquad1}
    x=\log\left(\frac{-c+\sqrt{c^2+4ab}}{2b}\right)=:g(a,b,c).
\end{align}
For our purposes, the case where $b=0$ is also of interest, and here
    \begin{align}\label{eq:logquad2}
    x=\log\left(\frac{a}{c}\right)=:g(a,0,c).
\end{align}
This follows from $ae^{-x}=c$, and of course, it also coincides with the value $\lim_{b\to 0}g(a,b,c)$.
\end{subequations}

Bringing all of the above together, we arrive at the following iterative algorithm,
where $\mu_j$ is computed as a function of $\nu$, using \eqref{eq:logquad}, to solve
\begin{subequations}\label{eq:SinkhornFB}
\begin{align}\label{eq:Sinkhorn1}   
&F_j(\mu_j) = a^F_j(\nu)e^{-\mu_j}-b^F_j(\nu)e^{\mu_j}=p_j,\\\nonumber &\mbox{ for }j=1:n,
\end{align}
and $\nu_i$ is computed as a function of $\mu$, using \eqref{eq:logquad} to solve
\begin{align}\label{eq:Sinkhorn2}
&B_i(\nu_i) = a^B_i(\mu)e^{-\nu_i}-b^B_i(\mu)e^{\nu_i}=1,\\\nonumber &\mbox{ for }i=1:n.
\end{align}
\end{subequations}
The steps are summarized below:
\begin{algorithm}[H]
\caption{Sinkhorn-like algorithm}\label{alg:1}
\begin{algorithmic}[1]
\STATE Initialize $\nu\in\mathbb R^n$, e.g., setting $\nu=0$.\\[0.045in]
\STATE 
For $j=1:n$, determine $\mu_j=g(a^F_j(\nu),b^F_j(\nu),p_j)$.\\[0.045in]
\STATE 
For $i=1:n$, determine $\nu_i=g(a^B_i(\mu),b^B_i(\mu),1)$.\\[0.045in]
\STATE Repeat steps 2 and 3 until convergence.
\end{algorithmic}
\end{algorithm}

When $A^-$ is the zero matrix, the above system of equations \eqref{eq:SinkhornFB} reduces to the classical Schr\"odinger system \cite{chen2021stochastic}. Indeed, in this case, $b^F,b^B$ vanish and the steps for finding $\mu_j$ and $\nu_i$, from $\nu$ and $\mu$, respectively, reduce to the standard diagonal scaling of the celebrated Sinkhorn algorithm, giving $e^{-\mu_j}=p_j/a_j^F(\nu)$, and similarly, $e^{-\nu_i}=1/a^B_i(\mu)$. Thus, the above algorithm represents a generalization of the Sinkhorn algorithm; the sign-indefiniteness of $A$ prevents simple diagonal scaling as an option to satisfy iteratively the boundary conditions, as in the original Sinkhorn iteration \cite{georgiou2015positive}, yet the update can still be done quite easily using \eqref{eq:logquad}.

Algorithm~\ref{alg:1} can also be seen as implementing {\em coordinate ascent} on the dual functional of the Lagrangian, with exact evaluation of the maximizer of a smooth functional along coordinate directions at each step. A detailed analysis together with extension of the general framework to address inverse problems in higher dimensions is the subject of forthcoming work~\cite{DGT}.

\subsection{Gradient descent method}

Once again, using duality as earlier, the dual function corresponding to the primal problem (Problem \ref{prob:prob2}) reads
\begin{align*}
g(\mu,\lambda)
= \!\!\!& \sum_{i,j|\mathbf{A}_{i,j}\neq 0} \!\!\!-\mathbf A_{ij}\exp(-\mu_jA_{ij} - \frac{\lambda_i}{p_i}A_{ij})p_{i}\\
&-\sum_j\mu_{j}p_{j} - \sum_i\lambda_i.
\end{align*}
Therefore, the dual problem reads
\begin{align*}
\max_{\mu,\lambda} \ \ g(\mu,\lambda),
\end{align*}
and can be also solved by gradient ascent. Specifically, taking partials of $g(\lambda,\mu)$ with respect to $\mu$ and $\lambda$, we have 
\begin{subequations}
\begin{align}
    \frac{\partial g}{\partial \mu_{j}}  &=  \frac{1}{e}\sum_{i} p_iA_{ij}\exp\big(- \mu_jA_{ij} -  \frac{\lambda_i}{p_i}A_{ij} \big) - p_j,\\
    \frac{\partial g}{\partial \lambda_{i}} &= \frac{1}{e}\sum_{j} A_{ij}\exp\big(- \mu_jA_{ij} -  \frac{\lambda_i}{p_i}A_{ij} \big)-1.
\end{align}
\end{subequations}
The optimizer in \eqref{eq:Pistar} is obtained by iteratively updating the dual variables in the direction of the gradient with a suitable step size, and is sketched below.

\begin{algorithm}[H]
\caption{Gradient descent method}\label{alg:2}
\begin{algorithmic}[1]
\STATE Initialize $\lambda=0$, $\mu= 0$, and select step size $\gamma$.
\STATE For $j=1:n$, determine
\[
 \mu_{j}^{\rm next} = \mu_{j} + \gamma(\frac{1}{e}\sum_{i} p_iA_{ij}\exp\big(- \mu_jA_{ij} -  \frac{\lambda_i}{p_i}A_{ij} \big) - p_j).
\]
\STATE For $j=1:n$, determine
\[
 \nu_{j}^{\rm next} = \nu_{j} + \gamma(\frac{1}{e}\sum_{i} A_{ij}\exp\big(- \mu_jA_{ij} -  \frac{\lambda_i}{p_i}A_{ij} \big) - 1).
\]
\STATE Update the values of $\mu,\nu$ and repeat until convergence.
\end{algorithmic}
\end{algorithm}

\subsection{Remarks}
Although the functional $\mathbf J$ in Problem~\ref{prob:prob2} is convex, existence of a minimizer hinges on whether the constraints are feasible. It is of interest to determine efficient ways to test feasibility, especially for very large matrices.

 It is also important to note that the theory applies to the case where $A$ is not necessarily symmetric and/or does not have a symmetric sign structure. A non-symmetric adjacency represents a directed graph, while asymmetry of the sign structure brings in interesting feedback dynamics.

\section{Illustrative examples}\label{sec:examples}
We now discuss representative examples to illustrate what can be accomplished with the proposed formulation. The code used for working out all of the examples may be found at {\em \url{https://github.com/dytroshut/negative-probability-forward-backward}}.

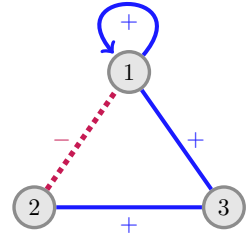
\begin{figure}[htb!]
\centering
\scalebox{.9}{
\begin{tikzpicture}
[p/.style={circle,draw=black!45,fill=black!10,line width=0.5mm,inner sep=0pt,minimum size=6mm}]
\node (n1) at (0,2) [p] {$1$};
\node (n2) at (-1.4,0) [p] {$2$};
\node (n3) at (1.4,0) [p] {$3$};

\path[line width=.6mm][blue!90] (n1) edge [in=130,out=50,loop] node[midway,below]{$\mathlarger{+}$} (n1);

\draw[line width=.8mm,dotted][purple!90] (n2) to [] node[midway,left](){$ \mathlarger{-}$} (n1);
\draw[line width=.6mm][blue!90] (n1) to [] node[midway,right](){$\mathlarger{+}$} (n3);
\draw[line width=.6mm][blue!90] (n3) to [] node[midway,below](){$\mathlarger{+}$} (n2);
\end{tikzpicture}}
\caption{Network topology}
\label{fig:example1}
\end{figure}
\begin{example}\label{ex:example1}
   Consider the $3$-node network shown in Figure \ref{fig:example1}.
 We take $A$ to be the following sign-indefinite adjacency matrix:
\begin{align*}
A = 
\begin{bmatrix}
\phantom{-}1 &-1 &\phantom{-}1\\
-1 &\phantom{-}0 &\phantom{-}1\\
\phantom{-}1  &\phantom{-}1 &\phantom{-}0
\end{bmatrix}.
\end{align*}
Note that the sign of the edge between the first two nodes is negative, indicating inhibiting effect. We take $p = [0.3 \ 0.3\ 0.4]$. Algorithms \ref{alg:1} and \ref{alg:2} converge to the same value (evidently), giving 
\begin{align*}
\Pi = 
\begin{bmatrix}
\phantom{-}0.7631 &-0.0482 &0.2581\\
-0.0482      &\phantom{-}0 &1.0482\\
\phantom{-}0.2138  &\phantom{-}0.7862 &0
\end{bmatrix},
\end{align*}
which can be verified to satisfy the constraints. As can be observed, $\Pi$ is not symmetric, as it was not required and only respects the signature structure of $A$.
\end{example}

\begin{example}\label{ex:example1a}
Once again we consider the $3$-node network shown in Figure \ref{fig:example1}. However, this time, we take $A$ to have the following sign-structure,
\begin{align*}
A = 
\begin{bmatrix}
\phantom{-}1  &\phantom{-}1 &\phantom{-}1\\
-1 &\phantom{-}0 &\phantom{-}1\\
\phantom{-}1  &\phantom{-}0 &\phantom{-}0
\end{bmatrix},
\end{align*}
which is already not symmetric. We take the same probability vector $p = [0.3 \ 0.3\ 0.4]$. The solution in this case can be computed by either algorithm to be 
\begin{align*}
\Pi = 
\begin{bmatrix}
\phantom{-}0.0009 &0.9962 &0.0029\\
-0.3321 &0 &1.3321\\
\phantom{-}1  &0 &0
\end{bmatrix},
\end{align*}
which can be verified to satisfy the constraints $\Pi^\prime p=p$ and $\Pi \mathds 1=\mathds 1$. This example highlights the fact that the theory applies equally well to the case where $A$ is not symmetric and/or does not have symmetric sign structure.
\end{example}

\begin{example}\label{sec:exp2}
We now work out an example with a substantially larger adjacency matrix $A$.
The sign indefinite adjacency matrix $A$ chosen for this example is
\begin{align*}
A = 
\begin{bmatrix}
\phantom{-}1  &\phantom{-}0 &\phantom{-}0 &\phantom{-}0 &-1 &\phantom{-}0 &\phantom{-}1 &\phantom{-}0 &\phantom{-}1 &\phantom{-}0\\
\phantom{-}0  &\phantom{-}0 &\phantom{-}1 &\phantom{-}0 &\phantom{-}0  &\phantom{-}0 &\phantom{-}0 &\phantom{-}0 &\phantom{-}0 &\phantom{-}1\\
\phantom{-}0  &\phantom{-}1 &\phantom{-}0 &\phantom{-}0 &\phantom{-}1  &\phantom{-}0 &\phantom{-}1 &\phantom{-}0 &\phantom{-}0 &\phantom{-}0\\
\phantom{-}0  &\phantom{-}0 &\phantom{-}0 &\phantom{-}0 &\phantom{-}1  &\phantom{-}0 &\phantom{-}0 &\phantom{-}1 &\phantom{-}0 &\phantom{-}0\\
-1            &\phantom{-}0 &\phantom{-}1 &\phantom{-}1 &\phantom{-}1  &-1 &\phantom{-}0 &\phantom{-}0  &\phantom{-}0 &\phantom{-}0\\
\phantom{-}0  &\phantom{-}0 &\phantom{-}0 &\phantom{-}0 &-1 &\phantom{-}0 &\phantom{-}1 &\phantom{-}0  &\phantom{-}1 &\phantom{-}0\\
\phantom{-}1  &\phantom{-}0 &\phantom{-}1 &\phantom{-}0 &\phantom{-}0  &\phantom{-}1 &\phantom{-}1 &\phantom{-}0  &\phantom{-}0 &\phantom{-}0\\
\phantom{-}0  &\phantom{-}0 &\phantom{-}0 &\phantom{-}1 &\phantom{-}0  &\phantom{-}0 &\phantom{-}0 &\phantom{-}0  &\phantom{-}1 &-1\\
\phantom{-}1  &\phantom{-}0 &\phantom{-}0 &\phantom{-}0 &\phantom{-}0  &\phantom{-}1 &\phantom{-}0 &\phantom{-}1  &\phantom{-}0 &\phantom{-}1\\
\phantom{-}0  &\phantom{-}1 &\phantom{-}0 &\phantom{-}0 &\phantom{-}0  &\phantom{-}0 &\phantom{-}0 &-1 &\phantom{-}1 &\phantom{-}0
\end{bmatrix}
\end{align*}
The probability vector $p$ is
\begin{align*}
p = \begin{bmatrix}
0.1 \ 0.05 \ 0.05 \ 0.15 \ 0.2 \ 0.05 \ 0.03 \ 0.07 \ 0.25 \ 0.05
\end{bmatrix}.
\end{align*}
The topology of the network is shown in Figure~\ref{fig:example2}, with color-coded display of the negative values using dashed red curves. The resulting matrix $\Pi$  is given in~\eqref{eq:examplePi}.
\end{example}

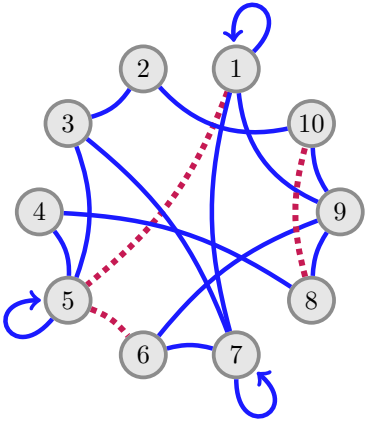
\begin{figure}[htb!]
\centering
\scalebox{1}{
\begin{tikzpicture}
[p/.style={circle,draw=black!45,fill=black!10,line width=0.5mm,inner sep=0pt,minimum size=6mm}]
\node (n1) at (0.6180,1.9021) [p] {$1$};
\node (n2) at (-0.6180,1.9021) [p] {$2$};
\node (n3) at (-1.6180,1.1756) [p] {$3$};			
\node (n4) at (-2,0) [p] {$4$};
\node (n5) at (-1.6180,-1.1756) [p] {$5$};
\node (n6) at (-.6180,-1.9021) [p] {$6$};
\node (n7) at (.6180,-1.9021) [p] {$7$};
\node (n8) at (1.6180,-1.1756) [p] {$8$};
\node (n9) at (2,0) [p] {$9$};
\node (n10) at (1.6180,1.1756) [p] {$10$};

\path[line width=.6mm][blue!90] (n1) edge [in=95,out=45,loop] node {} (n1);
\path[line width=.6mm][blue!90] (n5) edge [in=-180,out=-130,loop] node {} (n5);
\path[line width=.6mm][blue!90] (n7) edge [in=-40,out=-90,loop] node {} (n7);

\draw[line width=.8mm,dotted][purple!90] (n1) to [bend left=15] node[midway,right](){} (n5);
\draw[line width=.6mm][blue!90] (n1) to [bend right=15] node[midway,right](){} (n7);
\draw[line width=.6mm][blue!90] (n1) to [bend right=30] node[midway,right](){} (n9);
\draw[line width=.6mm][blue!90] (n2) to [bend left=20] node[midway,right](){} (n3);
\draw[line width=.6mm][blue!90] (n2) to [bend right=30] node[midway,right](){} (n10);
\draw[line width=.6mm][blue!90] (n3) to [bend left=20] node[midway,right](){} (n5);
\draw[line width=.6mm][blue!90] (n3) to [bend left=15] node[midway,right](){} (n7);
\draw[line width=.6mm][blue!90] (n4) to [bend left=20] node[midway,right](){} (n5);
\draw[line width=.6mm][blue!90] (n4) to [bend left=15] node[midway,right](){} (n8);
\draw[line width=.8mm,dotted][purple!90] (n5) to [bend left=15] node[midway,right](){} (n6);
\draw[line width=.6mm][blue!90] (n6) to [bend left=15] node[midway,right](){} (n7);
\draw[line width=.6mm][blue!90] (n6) to [bend left=15] node[midway,right](){} (n9);
\draw[line width=.6mm][blue!90] (n8) to [bend left=15] node[midway,right](){} (n9);
\draw[line width=.8mm,dotted][purple!90] (n8) to [bend left=15] node[midway,right](){} (n10);
\draw[line width=.6mm][blue!90] (n9) to [bend left=15] node[midway,right](){} (n10);
\end{tikzpicture}}
\caption{The topology of a $10$-node network}
\label{fig:example2}
\end{figure}

\begin{figure*}[htb!]
\begin{equation}\label{eq:examplePi}
\Pi = 
\begin{bmatrix}
\phantom{-}0.4337  &0   &0   &0   &{\purple-0.3398}  &\phantom{-}0  &0.1205  &\phantom{-}0   &0.7857  &\phantom{-}0\\
\phantom{-}0   &0  &0.2945  &0   &\phantom{-}0   &\phantom{-}0   &0  &\phantom{-}0    &0    &\phantom{-}0.7055\\
\phantom{-}0    &0.4595   &0   &0    &\phantom{-}0.4150   &\phantom{-}0    &0.1255  &\phantom{-}0  &0  &\phantom{-}0\\
\phantom{-}0   &0   &0   &0    &\phantom{-}0.6451   &\phantom{-}0   &0    &\phantom{-}0.3549   &0    &\phantom{-}0\\
{\purple-0.1930}   &0   &0.1623  &0.5895  &\phantom{-}0.6441  &{\purple-0.2029}   &0   &\phantom{-}0   &0  &\phantom{-}0\\
\phantom{-}0   &0    &0    &0   &{\purple-0.2470}    &\phantom{-}0    &0.1658   &\phantom{-}0    &1.0812    &\phantom{-}0\\
\phantom{-}0.4064    &0    &0.0940    &0    &\phantom{-}0    &\phantom{-}0.3866    &0.1129   &\phantom{-}0  &0   &\phantom{-}0\\
\phantom{-}0     &0     &0    &0.4586   &\phantom{-}0    &\phantom{-}0    &0   &\phantom{-}0    &0.9887   &{\purple-0.4473}\\
\phantom{-}0.3321    &0    &0    &0    &\phantom{-}0    &\phantom{-}0.3159  &0    &\phantom{-}0.1678    &0    &\phantom{-}0.1841\\
\phantom{-}0   &0.5405    &0  &0    &\phantom{-}0    &\phantom{-}0   &0   &{\purple-0.5038}    &0.9633   &\phantom{-}0  
\end{bmatrix}
\end{equation}
\hrulefill
\vspace*{4pt}
\end{figure*}

\begin{example}\label{sec:num_large}
We finally consider a random graph with $100$ nodes and $1428$ edges, shown in Figure \ref{fig:example}. The graph has $1360$ edges with positive signs and the remaining with negative. The edges with negative transition probability are randomly assigned and colored in red, and the ones with positive probability are shown in blue. Figure \ref{fig:convergence} displays the convergence of the Sinkhorn-like algorithm \ref{alg:1} in terms of the objective function $\mathbf J(\mathbf \Pi,\mathbf A)$ and marginal constraint violation\footnote{Measure used to assess convergence of the standard Sinkhorn iteration.} $\log(\|\Pi^\prime p-p\|)$. A linear convergence rate on the constraint/marginal violation is observed, and can be shown by pointing to the fact that Algorithm \ref{alg:1} can be seen as a coordinate ascent algorithm. The example can be replicated using code and data in the project's website.
\end{example}

\begin{figure}[H]
\centering
\includegraphics[width=0.9\columnwidth,trim={4.2cm 3.2cm 3.5cm 3cm},clip]{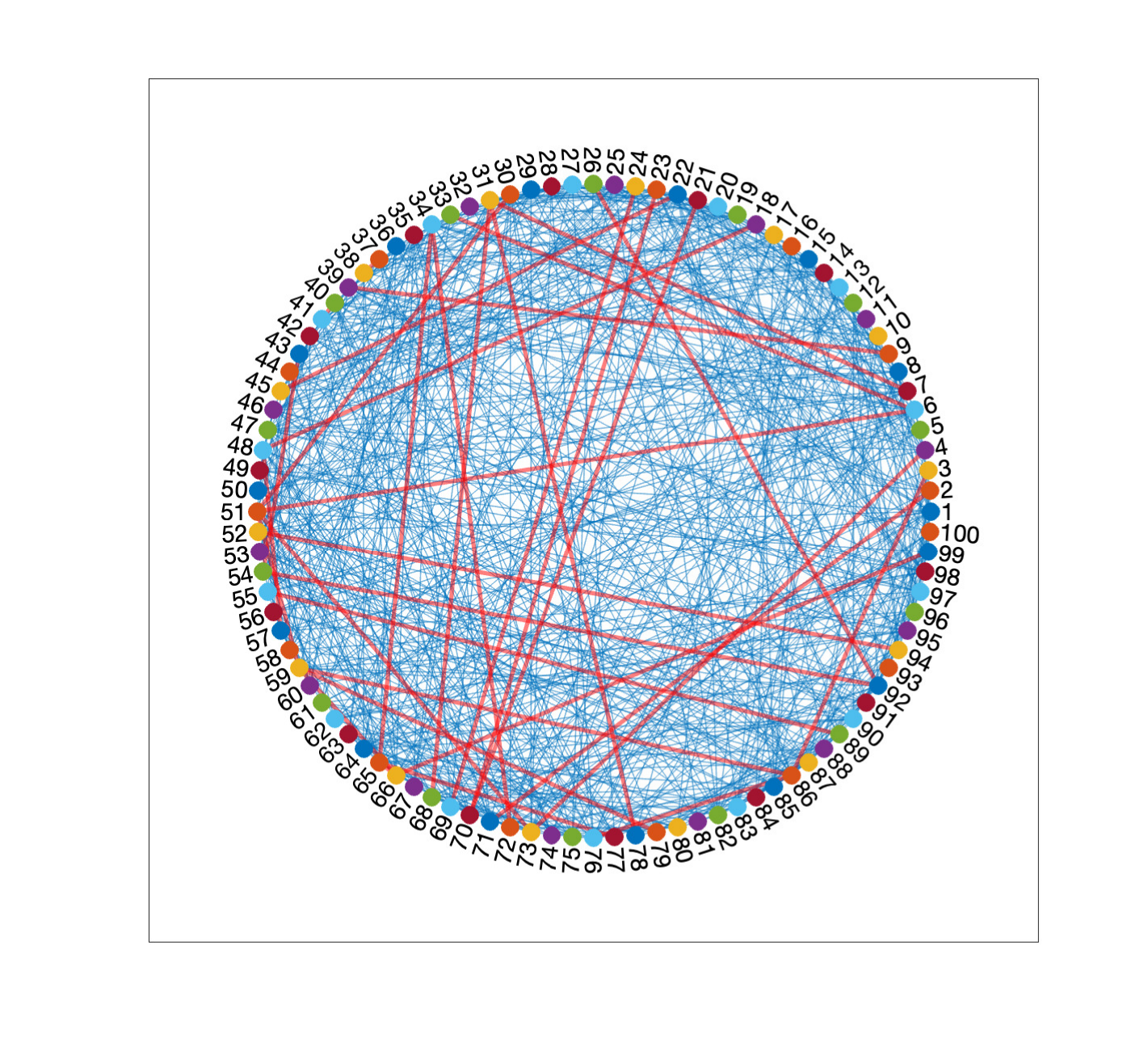}
\caption{The topology of the 100-node random network used. Nodes are enumerated and arranged sequentially on a circle. Edges assigned negative transition probability are shown in red, while edges corresponding to positive transition probability are shown in blue.}
\label{fig:example}
\end{figure}

\begin{figure}[H]
\centering
\includegraphics[width=0.85\columnwidth,trim={0.6cm 0cm 0.6cm 0cm},clip]{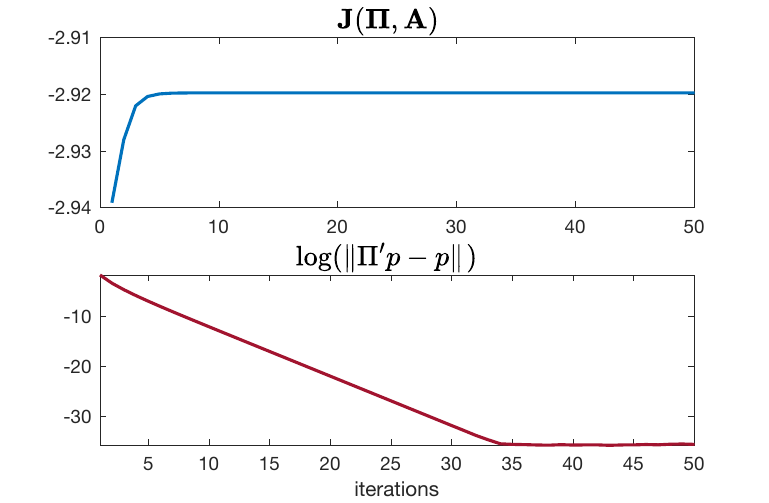}
\caption{The top subplot shows $\mathbf{J}(\mathbf{\Pi}, \mathbf{A})$ as a function of iteration number. The second subplot shows the violation/error of the marginal constraint $|\Pi^\prime p - p|$ in log-scale.}
\label{fig:convergence}
\end{figure}

\section{Concluding remarks}
The purpose of the present work has been to introduce a sign-indefinite probabilistic model
to reflect promoting/inhibiting affinity between nodes, as in co-expression of genes in gene regulatory networks. The theory is of independent interest and leads to a nonstandard extension of the classical Sinkhorn algorithm. An important question that remains is on how to efficiently test feasibility of Problems \ref{prob:prob1} and \ref{prob:prob2}. In addition, in light of the fact that the Sinkhorn iteration addresses a static version of the more general Schr\"odinger bridge problem \cite{chen2021stochastic}, it is of interest to explore a setting that allows modeling evolution of nodal mass, allowing for promotion/inhibition effects, under suitable dynamics.

\bibliographystyle{IEEEtran}
\bibliography{references}

\end{document}